\documentclass[11pt]{article}

\usepackage[utf8]{inputenc}
\usepackage[T1]{fontenc}
\usepackage{amsmath,amssymb}
\usepackage{graphicx}
\usepackage{amsthm}
\theoremstyle{plain}
\newtheorem{theorem}{Theorem}
\newtheorem{lemma}{Lemma}
\newtheorem{proposition}{Proposition}
\theoremstyle{definition}
\newtheorem{definition}{Definition}
\theoremstyle{remark}
\newtheorem{remark}{Remark}
\usepackage[numbers]{natbib}
\usepackage{hyperref}

\title{An Upper Bound on the Probability That a User Encounters an Undiscovered Defect}
\author{Carlos M. Hern\'andez-Su\'arez \and Karla Hern\'andez-Cuevas \\[4pt]
  Coordinaci\'on General de Investigaci\'on Cient\'ifica \\
  Universidad de Colima, Colima, Colima 28040, Mexico \\[2pt]
  \texttt{cmh1@cornell.edu} \quad \texttt{karlahdezcvs@gmail.com}}
\date{\today}

\begin{document}
\maketitle

\begin{abstract}
Before releasing software to a general population, a developer must weigh a
single question: if we ship now, what fraction of users will still hit a defect?
This is not a question about how many defects remain, nor whether any particular
defect is present---the quantities the reliability literature has long
estimated---but about a different and, for a release decision, more consequential
one: the probability that a user encounters a defect at all. We give a direct,
distribution-free answer. Reading each beta-test report as a draw from the user
population and each distinct defect as a class, we show that the fraction of
defects reported exactly once, $s/n$, is a conservative upper bound on the
probability that a user encounters a defect unseen in testing. This bound is the
exact maximum-likelihood estimate of the mass of unseen defects under a general
urn construction---the canonical form---into which any population of classes
embeds; because that construction charges every singleton to the unseen
reservoir, $s/n$ overstates the user's risk rather than understating it, the
direction a release decision requires. The estimate needs no operational profile,
no assumption on the number or frequency of defects, and no model of the
program's internal structure---since a defect's report count already reflects how
many users reach it, the estimate is invariant to whether the reachability graph
is a tree or a directed acyclic graph, and defects hidden behind other defects
are bounded automatically. We validate the estimator against synthetic
populations with known ground truth, and discuss the encounter-level data---beta
or crash telemetry---under which the user-facing reading holds.
\end{abstract}

\section{Introduction}

Software intended for release to a general population is hardened through
successive rounds of beta testing: users exercise the system, report the
failures they encounter, developers repair them, and the cycle repeats. Each
round poses the same decision---release now, or test another round?---and that
decision turns on a single quantity: the probability that a user of the released
software will still encounter a bug. A guarantee that this probability is small,
or more precisely an \emph{upper bound} on it, is what a principled release
decision requires.

Most existing methods answer a different question: how many defects remain.
Software reliability growth models fit the cumulative failure history to a
parametric process---classically the nonhomogeneous Poisson model of
\citet{goel1979}---and extrapolate to a total defect count, reporting the
difference from the number already found. Capture--recapture estimators borrowed
from ecology use the overlap among independent inspectors to estimate fault
content \citep{petersson2004}, an approach later driven by post-release user
reports rather than inspectors \citep{bucholz2009}. A third line converts a
prior estimate of residual defects into a conservative worst-case bound on the
failure rate \citep{bishop2002,salako2025}.

Two assumptions run through this body of work. Its estimand is a \emph{count} of
faults, and its passage from that count to a user-facing reliability figure
requires an \emph{operational profile}---the distribution over how the software
is exercised in the field---a quantity that is itself notoriously hard to
estimate \citep{bishop2002}. But the release decision is about users, not faults.
Ten defects reachable only through an obscure configuration menu threaten fewer
users than a single defect on the sign-in screen; a defect count is blind to this
difference, whereas a probability of user-facing failure is defined by it.
\citet{bucholz2009} confront the same gap directly, conceding that their
estimator is silent on defects that are never reported because they seldom
manifest, and fall back on an informal averaging argument.

A separate tradition estimates precisely a probability. The coverage, or
missing-mass, estimator of \citet{good1953}, credited by its author to Turing,
equates the probability that the next observation belongs to a previously unseen
class with $s/n$, the fraction of classes seen exactly once in a sample of size
$n$. The estimator is distribution-free and thoroughly characterized: it is
asymptotically optimal \citep{orlitsky2003,orlitsky2015}, concentrates sharply
around its target in finite samples \citep{mcallester2003,skorski2021}, admits
distribution-free high-probability upper bounds \citep{bubeck2013} and a normal
limit law for the coverage \citep{esty1983}, and is subject to a sharp
limitation---the missing mass cannot be consistently estimated without
assumptions on the tail of the class frequencies \citep{mossel2019}. The same
principle grounds class-number estimation in ecology \citep{chao1992} and, in
earlier work of one of the authors, the coverage of a conserved germplasm
collection \citep{hernandez2018}. Its estimand---the chance of meeting something
not yet seen---is exactly the release question, and it requires no operational
profile, dispensing with the assumption the count-based methods find most
troublesome.

Neither literature, then, answers the release question as posed. The
count-based methods estimate the wrong object and require an operational profile
they cannot reliably obtain; the coverage methods estimate the right object but
have entered software only through fuzzing, where the sampler is a machine rather
than a user \citep{bohme2018}, and were otherwise developed for flat,
exchangeable populations. In particular, no existing method bounds the
probability that a released \emph{user} reaches an undetected defect. This paper
supplies exactly that, by carrying a coverage estimator whose sampling
distribution is the user population.

The one place coverage has reached software is fuzzing, and the contrast with
that work makes our estimand precise. \citet{bohme2018} models automated test
generation as species discovery and uses the same estimator $s/n$ to bound the
residual risk that a fuzzing campaign has missed a bug, and hence to decide when
to stop the campaign. The estimator is shared; the sampling distribution is not,
and the sampling distribution is where the meaning lies. For a fuzzer, the
abundance of a bug is the probability that the tool's input generator stumbles
onto it---an artefact of the machine's randomness that says nothing about how the
software is used. In the release problem the sampler is the user population, and
the abundance $p_i$ of a bug is the fraction of users who encounter it. The two
frameworks thus evaluate the same $s/n$ over different measures---machine input
space in one, user behaviour in the other---and only the latter weights each bug
by how people actually exercise the software. Consequently $s/n$ here bounds the
proportion of \emph{users} who will meet a not-yet-seen bug, a user-facing
quantity the fuzzing framework cannot produce because it contains no users. This
reading holds provided the observations are genuine user encounters---for
instance beta-channel telemetry recording how many users each defect
affects---rather than a curated list with a single entry per bug.

Transplanting coverage into software meets an obstacle with no counterpart in
ecology. When one samples a biological population, every class is reachable on
every draw; a species is missed only by chance. Software is structured, not flat.
Its branching architecture is a building of rooms behind doors, and a bug is a
closed door: the code it guards cannot run, and any bugs residing beyond it
cannot surface, until that bug is fixed. Defects are therefore not exchangeable
draws from a fixed pool but are partially ordered by reachability, and the
population of \emph{discoverable} defects itself enlarges whenever a gating bug
is removed. A coverage figure computed on the current build appears, on its face,
to certify nothing about these nested and still unreachable defects---the
difficulty this paper must overcome.

The obstacle is, in fact, the resolution. A closed door is reachable only by the
users who first pass through it, so the total probability mass of everything
gated behind it cannot exceed the mass of the door itself; nesting adds nothing,
because reachability compounds multiplicatively down the tree and every factor is
at most one. A user-weighted coverage estimate at the current frontier therefore
already dominates the combined contribution of all undiscovered defects behind
it, gated or not, and the depth and width of the unexplored subtree need never be
measured.

\paragraph{Contributions.}
We make this argument precise and put it to work. First, we recast the release
question as the estimation of a user-weighted missing mass---the probability that
a released user reaches a defect not seen in testing---in place of a residual
defect count. Second, building on the urn construction of
\citet{hernandez2018}, we show that under the canonical form $s/n$ is an exact
maximum-likelihood estimate of that mass and, because the construction charges
every singleton to the unseen reservoir, a conservative upper bound on it. Third,
we show that the estimate depends on the software's reachability structure only
through the report counts, so it is invariant to whether that structure is a tree
or a directed acyclic graph, and defects hidden behind other defects are bounded
automatically. Fourth, we validate the estimate against synthetic populations
with known ground truth, and identify the encounter-level data under which the
user-facing reading holds empirically.

\section{The Canonical Form}
\label{sec:canonical}

Throughout, a \emph{population} is a probability distribution over a set of
classes---equivalently colours, types, or species---and $X_1,\dots,X_n$ is an
independent sample of size $n$ drawn from it. For the sample, let $m_r$ denote
the number of classes observed exactly $r$ times, so that $s:=m_1$ is the number
of \emph{singletons} and $n=\sum_r r\,m_r$. Our target is the \emph{unseen
proportion}
\[
  U \;=\; \sum_{c\,:\,c\ \text{unobserved}} p_c ,
\]
the total probability carried by classes absent from the sample; equivalently,
$U$ is the probability that the next draw belongs to a class not yet seen. In the
release problem of the Introduction, $U$ is the probability that a fresh user
exercises a defect not encountered in testing.

\subsection{The canonical urn}

We first isolate a family of populations rich enough to contain every case of
interest, yet simple enough to admit an exact likelihood analysis. The
construction and the maximum-likelihood result that follows are due to
\citet{hernandez2018}, in the setting of germplasm collections; we recall them
here and apply them to bounding user-facing release risk.

\begin{definition}[Canonical form]
\label{def:canonical}
A population is in \emph{canonical form} with $K$ components if its classes
divide into
\begin{itemize}
  \item $K-1$ \emph{ordinary} classes, the $i$-th carrying probability $p_i$ and
        contributing a single colour; and
  \item one \emph{reservoir} of probability
        $p_K = 1-\sum_{i=1}^{K-1}p_i$, composed of arbitrarily many classes of
        equal, vanishing probability, so that any two draws from the reservoir
        return distinct colours almost surely.
\end{itemize}
We call the reservoir the \emph{all-different} component: every draw from it
yields a colour returned by no other draw.
\end{definition}

The reservoir is the device that lets a single population account at once for the
colours already pinned down and for the endless supply of colours not yet seen.
Its necessity is not a matter of convenience, as the next two subsections show.

\subsection{Every population is an instance of the canonical form}

The word \emph{subset} is tempting but points the wrong way: the canonical family
is the general object, and the familiar models sit inside it as special cases.

\begin{proposition}[Universality]
\label{prop:universality}
Fix any population and any sample of size $n$ drawn from it, and let the classes
observed two or more times have counts $n_1,\dots,n_{K-1}$ while $s$ classes are
observed exactly once. Then the sample has positive probability under the
canonical population whose ordinary classes are the multiply-observed ones and
whose reservoir supplies the singletons. Moreover the classical multinomial model
over the observed classes is exactly the canonical restriction $p_K=0$.
Consequently the canonical form specialises to the classical model and extends it
to populations from which unseen classes continue to arise.
\end{proposition}

\begin{proof}
A class observed two or more times repeats a fixed colour and therefore populates
an ordinary component. A class observed exactly once is, from the sample alone,
indistinguishable from a fresh draw of the all-different reservoir, and may be
assigned to it; the sample then has positive probability under the stated
canonical population. Forbidding unseen classes is the constraint $p_K=0$, which
returns the ordinary multinomial over the observed colours.
\end{proof}

\subsection{Maximum-likelihood estimate of the unseen proportion}

Why the reservoir must be all-different, rather than a collection of ordinary
classes, is settled by a likelihood comparison.

\begin{lemma}[All-different dominance]
\label{lem:dominance}
Consider $s$ classes each observed once, and suppose a total probability $Q$ is
available to model them, the remaining classes held fixed. Assigning this mass to
$s$ ordinary classes contributes at most $(Q/s)^{s}$ to the sample likelihood,
whereas assigning it to the all-different reservoir contributes $Q^{s}$. The
reservoir is therefore the more likely explanation by a factor $s^{s}$, strictly
so for $s\ge 2$.
\end{lemma}

\begin{proof}
Model the $s$ singletons as ordinary classes with probabilities $q_1,\dots,q_s$
summing to $Q$. Each is seen once, so their joint contribution to the likelihood
is $\prod_{j=1}^{s} q_j$, which by the inequality of arithmetic and geometric
means is at most $(Q/s)^{s}$, with equality at $q_j=Q/s$. Model them instead as
draws from the all-different reservoir of mass $Q$. Because a first-appearance
class is identified only by its novelty and the reservoir's colours are
exchangeable, the $s$ reservoir draws contribute $Q^{s}$. The ratio is
$Q^{s}/(Q/s)^{s}=s^{s}$, which exceeds $1$ for $s\ge 2$.
\end{proof}

\begin{theorem}[Maximum-likelihood estimate of the unseen proportion]
\label{thm:mle}
Under the canonical form, the maximum-likelihood estimate of the unseen
proportion $U=p_K$ is
\[
  \widehat{U} \;=\; \widehat{p}_K \;=\; \frac{s}{n},
\]
with $\widehat{p}_i = n_i/n$ for the ordinary classes. In particular
$\widehat{U}=0$ when the sample contains no singletons.
\end{theorem}

\begin{proof}
By Lemma~\ref{lem:dominance} the likelihood is maximised by attributing every
multiply-observed class to an ordinary component and every singleton to the
reservoir: a class seen twice or more cannot arise from the all-different
reservoir, and a class seen once is more likely there. The resulting likelihood
of the sample is, up to a factor independent of the parameters,
\[
  L(p_1,\dots,p_{K-1},p_K)\;\propto\;
  \Big(\prod_{i=1}^{K-1} p_i^{\,n_i}\Big)\,p_K^{\,s},
  \qquad \sum_{i=1}^{K-1}p_i + p_K = 1 ,
\]
a multinomial likelihood in $K$ cells with observed counts
$(n_1,\dots,n_{K-1},s)$ and total $n=\sum_i n_i + s$. Maximising subject to the
constraint, by a Lagrange multiplier, assigns each cell its empirical frequency:
$\widehat{p}_i=n_i/n$ and $\widehat{p}_K=s/n$. When $s=0$ the reservoir cell is
empty and $\widehat{p}_K=0$.
\end{proof}

\begin{remark}[The estimate is a conservative upper bound]
\label{rem:upperbound}
Under the canonical form the unseen proportion equals the reservoir mass, which
Theorem~\ref{thm:mle} estimates by $s/n$. The canonical form reaches this value
by charging \emph{every} singleton to the all-different reservoir. A real
population need not oblige: a class observed once may be an ordinary class that
merely happened to appear a single time, in which case its mass belongs to a
seen class, not to the reservoir. Attributing every singleton to the reservoir
can therefore only inflate the estimate, and $s/n$ overstates the true unseen
proportion on average,
\[
  \mathbb{E}\!\left[\tfrac{s}{n}\right] \;\ge\; \mathbb{E}[U],
\]
so that $s/n$ is a one-sided, upper estimate of $U$ rather than a point
prediction \citep{hernandez2018}. The direction is the one a release decision
requires: the estimate errs toward overstating the probability that a user meets
an unseen bug, never understating it. The mechanism is plainest in the extreme
case of a fully sampled population, in which every class has been seen and the
true unseen mass is zero, yet a class observed exactly once still contributes to
$s/n$; the estimate is then pure over-attribution.
\end{remark}

\begin{remark}
The canonical form is not offered as an alternative route to the number $s/n$,
which is already the Good--Turing coverage estimate $1-\widehat{C}$
\citep{good1953}. Its value is that it carries the estimate past a point at which
earlier treatments explicitly stopped. Over the classical model, whose only
classes are those observed, the maximum-likelihood estimate of the unseen mass is
identically zero---the method cannot see past the sample. Seeking a genuinely
distribution-free likelihood treatment, \citet{good1953} instead recorded that it
``leads to a mathematical problem that I have not solved,'' and fell back on an
expected-frequency argument with smoothing; the underlying difficulty was later
shown to be fundamental for \emph{consistent} estimation of the missing mass
\citep{mossel2019}. Theorem~\ref{thm:mle} passes the first of these walls: under
the canonical form, $s/n$ is not an approximation but the exact
maximum-likelihood estimate of the unseen proportion---an MLE, we stress, rather
than a claim of consistency, so no conflict with the impossibility result arises.
The same construction, moreover, is what will let us establish the estimate's
behaviour under the reachability constraints of software.
\end{remark}

\begin{remark}[Toward the release problem]
Under the canonical form, the probability that the next observation belongs to an
unseen class is exactly the reservoir mass $p_K$. In the release setting the next
observation is the next user's execution, and $p_K$ is the probability that this
user reaches a defect class unseen in testing. It remains only to confirm that
this reading does not depend on the internal structure of the software, to which
we turn in Section~\ref{sec:scope}.
\end{remark}

\subsection{Scope and the role of the architecture}
\label{sec:scope}

The estimate is a snapshot: a function of the reports in hand at a given moment,
returning an upper bound on the proportion of users who would still meet a bug
were the software released in its current state. Whether to release, or to gather
further reports, is the developer's decision and lies outside the method; if more
reports are collected, $s/n$ is simply recomputed on the enlarged sample. We
claim nothing about how one estimate relates to the next, only that each is valid
for the state on which it is computed.

\paragraph{Why the architecture does not matter.}
The architecture decides one thing: how many users reach each defect. A defect on
a common path is reached by many, one behind a closed door by none. Whether the
paths form a tree, in which each defect is reached by a single route, or a
directed acyclic graph (DAG), in which a defect may be reached by several, only
changes these numbers---it does not change what they are.

And those numbers are exactly what the reports record: a defect reached by many is
reported by many, one reached by none is reported by no one. So the architecture's
whole effect is already in the counts.

Our estimate is computed from the counts alone. It never consults the structure
because it does not need to---the structure has already left its mark there. Tree
or DAG, the same reports give the same $s/n$.

\section{A synthetic demonstration}
\label{sec:synthetic}

Real defect archives do not directly expose the quantity we wish to bound: a
bug report records that a defect was filed, not how many users encountered it,
and the software from which the reports arise changes over time. We therefore
validate the estimator on synthetic populations, where the ground truth is known
and the estimate can be checked against it directly.

\paragraph{Setup.}
We fix a population of $M$ bug classes with reach probabilities
$p_1,\dots,p_M$, where $p_i$ is the fraction of users who encounter bug $i$. A
sample of size $n$ is drawn---each draw a user encounter---and we record the
estimate $s/n$ together with the \emph{true} unseen proportion
$U=\sum_{i:\,c_i=0} p_i$, the total reach of the bugs absent from the sample,
which is available because we chose the population. All quantities below are
Monte-Carlo averages over many replications. The estimator is never told $M$ or
the shape of $\{p_i\}$; it sees only the sample.

\paragraph{The estimate is one-sided.}
Table~\ref{tab:synthetic} reports a population of $20$ bugs with $p_i=1/20$, over
a range of sample sizes. In every row the mean estimate meets or exceeds the
true missing mass, confirming $\mathbb{E}[s/n]\ge\mathbb{E}[U]$ of
Remark~\ref{rem:upperbound}. The gap is small---the estimate is accurate, not
merely conservative---but it never changes sign. Figure~\ref{fig:synthetic}
shows the same behaviour for a heavy-tailed population: the $s/n$ curve lies on
or above the true missing mass throughout, and the shaded region is the
over-attribution.

\begin{table}[h]
\centering
\begin{tabular}{rcccc}
\hline
$n$ & $\mathbb{E}[s/n]$ & $\mathbb{E}[U]$ & gap & $\Pr(s/n \ge U)$ \\
\hline
 25 & 0.2928 & 0.2766 & $+0.0162$ & 0.53 \\
 50 & 0.0809 & 0.0765 & $+0.0043$ & 0.57 \\
100 & 0.0063 & 0.0059 & $+0.0004$ & 0.89 \\
200 & 0.0000 & 0.0000 & $+0.0000$ & 1.00 \\
400 & 0.0000 & 0.0000 & $+0.0000$ & 1.00 \\
\hline
\end{tabular}
\caption{A population of $20$ bugs with $p_i=1/20$. The estimate $s/n$ meets or
exceeds the true unseen proportion $U$ at every sample size; the bound is
one-sided.}
\label{tab:synthetic}
\end{table}

\paragraph{The over-attribution made visible.}
The mechanism of Remark~\ref{rem:upperbound} is starkest when the population is
nearly saturated. With $20$ bugs and $n=100$ encounters, all $20$ bugs are seen
in about $89\%$ of samples, so the true missing mass is essentially zero; yet
$s/n$ averages about $0.006$, and among the fully covered samples---where the
true unseen mass is exactly zero---it remains positive. That residual is pure
over-attribution: a known bug, observed once, charged by the canonical form to
the reservoir of the unseen. The estimator cannot tell such a singleton from a
genuinely rare one, and errs, as intended, on the side of caution.

\begin{figure}[h]
\centering
\includegraphics[width=0.82\linewidth]{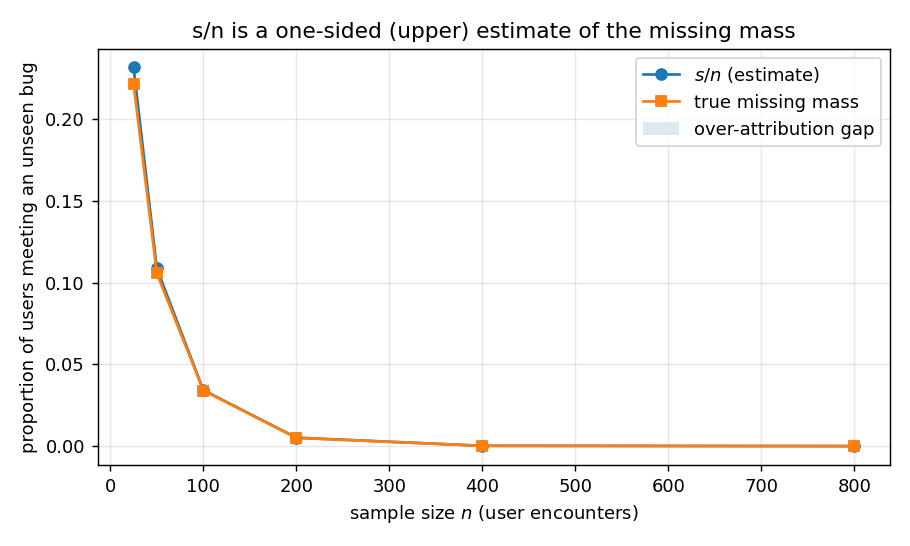}
\caption{A heavy-tailed population of $20$ bugs. The estimate $s/n$ (circles)
lies on or above the true missing mass (squares) at every sample size; the
shaded region is the over-attribution gap. Both decline as the sample grows,
the release story in miniature: continue testing until the bound, not merely the
estimate, falls below the tolerated risk.}
\label{fig:synthetic}
\end{figure}

\section{Discussion}
\label{sec:discussion}

The release question---what fraction of users will still meet a bug if we ship
now?---admits a direct, distribution-free answer. Reading each beta report as a
draw from the user population and each distinct defect as a class, the singleton
fraction $s/n$ estimates the probability that the next user encounters a defect
unseen in testing. Because the canonical form charges every singleton to the
reservoir of the unseen, the estimate is one-sided: it overstates that
probability rather than understating it, the direction a release decision needs.
It requires no operational profile, no assumption on the number or frequency of
defects, and no model of the software's internal structure, since the report
counts already carry the effect of that structure.

The estimator itself is not new; it is the coverage estimator of
\citet{good1953}, given a maximum-likelihood foundation through the canonical
urn by \citet{hernandez2018}. What is new here is its estimand and its domain: a
bound on the proportion of \emph{users} affected, for the release decision. This
separates the method from the two bodies of work with which it might be
confused. The software-reliability literature---growth models and
capture--recapture alike---estimates a defect \emph{count}, and needs an
operational profile to turn that count into a user-facing figure; we estimate the
user-facing figure directly. And while \citet{bohme2018} brought the same
coverage estimator to software, its sampler is a fuzzer exploring an input space,
so its abundances reflect a machine's randomness; ours reflect how a population
of people uses the software. A fuzzer and a user base differ not only in what
they sample but by orders of magnitude in how much: fuzzing campaigns saturate
coverage with millions of inputs and drive $s/n$ toward zero, whereas a beta
population is sparse---and it is precisely in the sparse regime that the bound is
informative.

The method's central requirement is also its central limitation. The user-facing
reading holds only when the observations are genuine user encounters, so that a
defect's multiplicity reflects how many users met it. Public defect archives do
not meet this condition: filing a report is costly, most defects are filed once
regardless of how many users hit them, and the reporting step erases the very
multiplicity the estimate depends on. Compounding this, an archive spanning many
releases is not a sample from one fixed population but a superposition of many,
so its coverage never settles. These are the reasons we validated the estimator
against synthetic populations with known ground truth rather than against an
archive. Realising the user-facing claim empirically calls for encounter-level
data---beta-channel or crash telemetry that records automatically how many users
each defect affects---where the reporting filter is absent by construction.

In use, the estimate is a snapshot, valid for the software as it stands; the
decision to ship or to keep testing is the developer's, and each further round of
testing simply yields a fresh $s/n$ on the enlarged sample. The natural rule is
to continue until the bound falls below the tolerated risk. Since $s/n$ is an
accurate but noisy estimate, one may also bound the sampling error, for which
distribution-free concentration inequalities for the missing mass are available
\citep{mcallester2003,bubeck2013}.

Several refinements suggest themselves. The estimate weights each defect only by
its reach, treating a crash on the sign-in screen and a cosmetic glitch alike; a
severity-weighted coverage would let the bound speak to user-facing harm rather
than mere incidence. The draws are assumed exchangeable, which correlated usage
patterns may violate. Finally, the canonical-urn maximum likelihood invites
comparison with profile maximum likelihood \citep{acharya2017}, which maximises
over the same enlarged parameter space by a different route; relating the two
precisely we leave to future work.

\bibliographystyle{plainnat}
\bibliography{refs}

\end{document}